\documentclass[11]{article}
\usepackage[T1]{fontenc}
\usepackage{blindtext}
\usepackage{listings}
\usepackage{fullpage}
\usepackage{amsfonts}
\usepackage{amssymb,amsmath,amsthm}
\usepackage{graphicx}
\usepackage[ruled, vlined,linesnumbered]{algorithm2e}
\usepackage{algpseudocode}
\usepackage{mathpazo,bbm}
\usepackage{xcolor}
\usepackage{url}

\usepackage{hyperref}  
\hypersetup{colorlinks=true,citecolor=blue,linkcolor=blue}
\definecolor{mygreen}{RGB}{28,172,0} 
\definecolor{mylilas}{RGB}{170,55,241}

\usepackage[numbered,framed]{matlab-prettifier}

\def\R{{\mathbb{R}}}
\def\eps{{\varepsilon}}

\newcommand{\diag}{\mathsf{diag}}

\newcommand{\vol}{\mathbf{vol}}

\DeclareMathOperator*{\E}{{\bf{E}}}
\newcommand{\Tr}{{\mathsf{Tr}}}

\newcommand{\de}{\mathsf{d}}
\newcommand{\nnz}{\mathsf{nnz}}

\newtheorem{theorem}{Theorem}[section]

\newtheorem{lemma}[theorem]{Lemma}

\newtheorem{definition}[theorem]{Definition}

\newtheorem{proposition}[theorem]{Proposition}

\usepackage{times}

\begin{document}
\title{A near-optimal algorithm for approximating the John Ellipsoid}
\author{%
Michael B. Cohen\thanks{\texttt{micohen@mit.edu}. MIT}
 \and
 Ben Cousins\thanks{\texttt{b.cousins@columbia.edu}. Columbia University}
 \and
 Yin Tat Lee\thanks{\texttt{yintat@uw.edu}. University of Washington}
 \and
 Xin Yang\thanks{\texttt{yx1992@cs.washington.edu}. University of Washington}
}

\begin{titlepage}
  \maketitle
  \begin{abstract}
\begin{abstract}
We develop a simple and efficient algorithm for approximating the John Ellipsoid of a symmetric polytope. 
Our algorithm is near optimal in the sense that our time complexity matches the current best verification algorithm. 
We also provide the MATLAB code for further research. 
\end{abstract}

  \end{abstract}
  \thispagestyle{empty}
\end{titlepage}


\section{Introduction}

Let $P=\{x\in \R^{n}:Ax\leq b\}$ be a polytope 
where $P$ has nonzero, finite Euclidean volume.
The classical theorem of~\cite{John48} states that if $E \subseteq P$ is the ellipsoid of maximal volume contained in $P$, then $P \subseteq nE$, where $nE$ represents a dilation of the ellipsoid $E$ by a factor of $n$ about its center. Moreover, if $P$ is symmetric, then $P \subseteq \sqrt{n}E$. 
The maximal volume inscribed ellipsoid (MVIE) $E$ is called the John Ellipsoid,
and we are interested in the problem of approximating $E$ when the polytope $P$ is centrally symmetric, i.e.\ $P$ can be expressed as $P = \{x \in \R^n: -\mathbf{1}_m \le Ax \le \mathbf{1}_m\}$ where $A \in \R^{m \times n}$ and $A$ has rank $n$.

The problem of computing the ellipsoid of maximal volume inside polytope given by a set of inequalities has a wealth of different applications, including sampling and integration~\cite{VempalaSurvey,chen2018fast}, linear bandits~\cite{bck12,hazan2016volumetric}, linear programming ~\cite{ls14}, cutting plane methods~\cite{Khachiyan88} and differential privacy~\cite{ntz13}.

Computing the John Ellipsoid additionally has applications in the field of experimental design, a classical problem in statistics~\cite{atwood1969optimal}. 
Specifically, the D-optimal design problem wants to maximize the determinant of the Fisher information matrix~\cite{kiefer1960equivalence,atwood1969optimal},
which turns out to be equivalent to finding the John Ellipsoid of a symmetric polytope.
While this equivalence is known, e.g.~\cite{Todd2016}, we include it in Section~\ref{sec:problem} for completeness. The problem of D-optimal design has received recent attention in the machine learning community, e.g.~\cite{azlsw17,wang2017computationally,lfn18}.
\subsection{Our Contribution}

Our main contribution is to develop an approximation algorithm to computing the John Ellipsoid inside a centrally symmetric polytope given by a set of inequalities.
Previously,
for solving the MVIE problem or its dual equivalent D-optimal design problem,
researchers have developed various algorithms,
such as first-order methods~\cite{kha96,ky05,dst08},
and second-order interior-point methods~\cite{nn94,sf04}.
Instead of using traditional optimization methods,
we apply a very simple fixed point iteration.
The analysis is also simple and clean,
yet the convergence rate is very fast.
We state our main result as follows.
\begin{theorem}[Informal]\label{thm:informal-no-jl}
	Given $A\in \R^{m\times n}$,
	let $P$ be a centrally symmetric polytope defined as $\{x\in \R^{n}: -\mathbf{1}_m\leq Ax\leq \mathbf{1}_m\}$.
	For $\eta\in (0,1)$,
	there is an algorithm (Algorithm \ref{alg:short-faster}) that runs in time $O(\eta^{-1} mn^2\log (m/n) )$,
	returning an ellipsoid $Q$ so that $\frac 1 {\sqrt{1+\eta}}\cdot Q\subseteq P\subseteq \sqrt{n}\cdot Q$.
\end{theorem}
In Lemma \ref{lem:approx_quality},
we show that our ellipsoid is $\eta$-close to the John Ellipsoid in a certain sense.
However,
if we want to get $(1-\epsilon)$-approximation to the maximal volume,
we shall set $\eta=\epsilon/n$\footnote{For details, see Lemma \ref{lem:approx_quality}.},
then Algorithm \ref{alg:short-faster} runs in time $O(\epsilon^{-1}mn^3\log (\frac m n))$,
and when $\epsilon$ is constant,
this is comparable with the best known results $O(mn^3/\epsilon)$~\cite{ky05,ty07}.

Furthermore,
we use sketching ideas from randomized linear algebra to speed up the algorithm so that the running time does not depend on $m$ explicitly.
This will make sense if $A$ is a sparse matrix.
Our result is stated as follows.
\begin{theorem}[Informal]\label{thm:informal-jl}
	Given $A\in \R^{m\times n}$,
	let $P$ be a centrally symmetric polytope defined as $\{x\in \R^{n}: -\mathbf{1}_m\leq Ax\leq \mathbf{1}_m\}$.
	For $\eta\in (0,1)$ and $\delta\in (0,1)$,
	there is an algorithm (Algorithm \ref{alg:faster}) that runs within $O(\frac {1} {\eta}\log \frac m {\delta})$ many iterations,
	returning an ellipsoid $Q$ so that with probability at least $1-\delta$,
	$\frac 1 {\sqrt{1+\eta}}\cdot Q\subseteq P\subseteq \sqrt{n}\cdot Q$. 
	Moreover,
	each iteration involves in solving $O(\frac 1 \eta)$ linear systems of the form $A^\top WAx=b$ where $W$ is some diagonal matrix.
\end{theorem}

Algorithm \ref{alg:faster} is near optimal,
because in order to verify the correctness of the result,
we need to compute the leverage scores of some weighted version of $A$.
The best known algorithm for approximating leverage scores
needs to solve $\tilde O(\frac 1 {\eta^2})$ many linear systems~\cite{ss11,dmmw12,cw13,nn13}.
One key advantage of our algorithm is that it reduces the problem of computing John ellipsoid to a relatively small number of linear systems. 
Therefore, it allows the user to apply the linear systems solver tailored for the given matrix $A$. 
For example, if $A$ is tall, one can apply sketching technique to solve the linear systems in nearly linear time \cite{w14}; 
if each row of A has only two non-zeros, one can apply Laplacian solvers \cite{ds08,kosz13,ckmpprx14,ks16,klpss16}. 
In the code we provided, we used the Cholesky decomposition which is very fast for many sparse matrices $A$ in practice.

\subsection{Related Works}

There is a long line of research on computing the maximal volume ellipsoid inside polytopes given by a list of linear inequalities. We note that \cite{kt93} presented a linear time reduction from the problem of computing a minimum volume enclosing ellipsoid (MVEE) of a set of points to the maximal volume inscribed ellipsoid problem; therefore, these algorithms also hold for approximating the John Ellipsoid.

Using an interior-point algorithm,
\cite{nn94} showed that a $1+\epsilon$ approximation of MVEE can be computed in time $O(m^{2.5}(n^2+m)\log (\frac{m}{\epsilon}))$. \cite{kt93} subsequently improved the runtime to $O(m^{3.5}\log (\frac{m}{\epsilon})\cdot \log (\frac{n}{\epsilon}))$.
Later on, \cite{ne99} and \cite{ans02} independently
obtained an $O(m^{3.5}\log \frac{m}{\epsilon})$ algorithm.
To the best of authors' knowledge,
currently the best algorithms by~\cite{ky05,ty07} run in time $O(mn^3/\epsilon)$.
We refer readers to \cite{Todd2016} for a comprehensive introduction and overview.

Computing the minimum volume enclosing ellipsoid of a set of points is the dual problem of D-optimal design. By generalizing smoothness condition on first order method, \cite{lfn18} managed to solve D-optimal design problem within $O(\frac{m}{\epsilon}\log (\frac n \epsilon))$ many iterations.
However, in the dense case,
their iteration costs $O(mn^2)$ time, which leads to larger running time comparing to \cite{ky05,ty07}.  
\cite{gp18} applied Bregman proximal method on the D-optimal design problem and observe accelerated convergence rate in their numerical experiments; however, they did not prove that their experimental parameter settings satisfy the assumption of their algorithm\footnote{We are grateful Gutman and Pe{\~n}a provided us their code for testing, of which D-optimal design was only one application.}.

A natural version of the D-optimal design problem is to require an integral solution. The integral variant is shown to be $\mathbf{NP}$-hard~\cite{vh12}, although recently approximation algorithms have been developed~\cite{azlsw17,sx18}. In our context, this means the weight vector $w \in \R^m$ is the integral optimal solution to \eqref{eq:ellipsoid-program-weights}, where the sum of the weights is some specified integral parameter $k$.

Several Markov chains for sampling convex bodies have well understood performance guarantees based upon the roundedness of the convex body. If $B_n \subseteq K \subseteq R \cdot B_n$, then the mixing time of hit-and-run and the ball walk are both $O(n^2 R^2)$ steps~\cite{LovaszVempala06, KannanLovaszSimonovits97}. Thus, placing a convex body in John position guarantees the walks mix in $O(n^4)$ steps, and $O(n^3)$ steps if the body is symmetric; this transformation is used in practice with the convex body to be a polytope~\cite{chrr2017}. Generating the John Ellipsoid, with a fixed center point, has also been employed as a proposal distribution for a Markov chain~\cite{chen2018fast,gustafson2018john}.

We build our even faster algorithm via sketching techniques.
Sketching has been successfully applied to speed up different problems,
such as linear programs \cite{lsz19}, clustering \cite{cemmp15,syz18}, low rank approximations \cite{cw13,nn13,bw14,cw15b,rsw16,swz17}, linear regression \cite{cw13,nn13,clmmps15,cw15a,psw17,alszz18}, total least regression \cite{dswy19},
tensor regression \cite{lhw17,dssw18} and tensor decomposition \cite{wtsa15,swz16,swz19}.
Readers may refer to \cite{w14} for a comprehensive survey on sketching technique.
We use sketching techniques to speed up computing leverage scores.
This idea was first used in \cite{ss11}.

Previous research on the MVEE problem did take advantage of the sparsity of the input matrix $A$,
and to the best of our knowledge,
our algorithm is the first one that is able to deal with large sparse input.
It would be interesting if we can apply sketching techniques to further speed up existing algorithms.

\subsubsection{Relation with \cite{cp15}}
We shall mention that our work is greatly inspired by \cite{cp15}.
The $\ell_p$ \emph{Lewis Weights} $\overline{w}$ for matrix $A\in \mathbb{R}^{m\times n}$ is defined as the unique vector $\overline{w}$ so that for $i\in [m]$,
\begin{align*}
a_i^\top \left(A^\top \diag(\overline{w})^{1-2/p}A\right)^{-1}a_i=\overline{w}_i^{2/p}.
\end{align*}
It is known that computing the $\ell_{\infty}$ Lewis Weight is equivalent to computing the maximal volume inscribed ellipsoid.
\cite{cp15} proposes an algorithm for approximating Lewis Weights for all $p<4$.
Their algorithm is an iterative algorithm that is very similar to our Algorithm \ref{alg:short-faster},
and the convergence is proved by arguing the iteration mapping is contractive.
The main difference is that \cite{cp15} outputs the weights in the last round,
while our Algorithm \ref{alg:short-faster} takes the average over all rounds and outputs the averaging weights,
which allows us to conduct a convexity analysis and deal with the $\ell_{\infty}$ case.

\section{Problem Formulation}\label{sec:problem}

In this section, we formally define the problem of computing the John Ellipsoid of a symmetric polytope.
Let $P=\{x\in \R^{n}:|a_i^{\top}x|\leq 1, i\in [m]\}$ be a symmetric convex polytope,
where $[m]$ denotes the set $\{1,2,\cdots,m\}$.
We assume $A=(a_1\, a_2\, \cdots \,a_m)^\top$ 
has full rank.
By symmetry, we know that the maximal volume ellipsoid inside the polytope should be centered at the origin.
Any ellipsoid $E$ centered at the origin can be expressed by $x^{\top}G^{-2}x\leq 1$,
where $G$ is a positive definite matrix.
Note that the volume of $E$ is proportional to $\det G$, and an
ellipsoid $E$ is contained in polytope $P$ if and only if for $i\in [m]$,
$\max_{x\in E} |a_i^{\top}x|\leq 1$.
For any $x\in E$,
we can write $x=Gy$ where $\|y\|_2\leq 1$.
Hence
\[
\max_{x\in E} |a_i^{\top}x|=\max_{\|y\|_2\leq 1}|a_i^{\top}Gy|= \max_{\|y\|_2\leq 1}\|Ga_i\|_2\cdot \|y\|_2=\|Ga_i\|_2
\]

Therefore, we can compute the John Ellipsoid of $P$ by solving the following optimization program:
\begin{equation}
\begin{aligned}\label{eq:ellipsoid-program-original}
\text{Maximize}  &\quad \log \det G^2,& \\
\text{subject to:}
&\quad G\succeq 0,&\\
&\quad \|Ga_i\|_2\leq 1,& \forall i\in [m].
\end{aligned}
\end{equation}

It turns out that the optimal ellipsoid satisfies $G^{-2}=A^{\top}\diag(w)A$,
where $w\in \mathbb{R}_{\geq 0}^m$ is the optimal solution of the program
\begin{equation}
\begin{aligned}\label{eq:ellipsoid-program-weights}
\text{Minimize}  &\quad \sum_{i=1}^m w_i-\log \det \left(\sum_{i=1}^m w_ia_ia_i^{\top}\right)-n,& \\
\text{subject to:}
&\quad w_i\geq 0, \quad \forall i\in [m].
\end{aligned}
\end{equation}

Actually program \eqref{eq:ellipsoid-program-weights} is the Lagrange dual of program \eqref{eq:ellipsoid-program-weights}.
Moreover,
we have the following optimality criteria for $w$ in the above program.
\begin{lemma}[Optimality criteria, Proposition 2.5 in \cite{Todd2016}]\label{lem:opt-criteria}
A weight $w$ is optimal for program \eqref{eq:ellipsoid-program-weights} if and only if 
\begin{align*}
&\sum_{i=1}^m w_i=n,&\\
&a_i^{\top}\left(\sum_{i=1}^m w_ia_ia_i^{\top}\right)^{-1}a_i=1, & \quad \mbox{if } w_i\neq 0;\\
&a_i^{\top}\left(\sum_{i=1}^m w_ia_ia_i^{\top}\right)^{-1}a_i<1, & \quad \mbox{if } w_i=0.
\end{align*}
\end{lemma}

Computing the John Ellipsoid is closely related to D-optimal design problem \cite{atwood1969optimal,BoydVandenbergheBook, Todd2016, gp18}.
For the D-optimal design problem,
we are given input $X\in \R^{n\times m}$ where $m>n$,
and we want to solve program,
\begin{equation}
\begin{aligned}\label{eq:d-optimal-design}
\text{Maximize}  &\quad \log \det \left(X\diag(v)X^\top\right),& \\
\text{subject to:}
&\quad v_i\geq 0, ~\forall i\in [m]\\
&\quad \sum_{i=1}^m v_i=1 
\end{aligned}
\end{equation}
We emphasize that program \eqref{eq:d-optimal-design} and program \eqref{eq:ellipsoid-program-weights} are equivalent, in the following sense.
By Lemma \ref{lem:opt-criteria},
we can rewrite program \eqref{eq:ellipsoid-program-weights} as
minimizing $n\log n-\log \det (A^\top \diag(w)A)$,
subject to $w_i\geq 0$ for $i\in [m]$ and $\sum_{i=1}^m w_i=n$.
By setting $v_i=\frac {w_i} n$,
we obtain program \eqref{eq:d-optimal-design}. Thus, optimal solutions to programs \eqref{eq:ellipsoid-program-weights} and \eqref{eq:d-optimal-design} are equivalent up to a multiplicative factor $n$.
 
We can also talk about an approximate John Ellipsoid.
\begin{definition}\label{def:weights-approx}
For $\epsilon>0$,
we say $w\in \mathbb{R}^m_{\geq 0}$ is a $(1+\epsilon)$-approximation of program \eqref{eq:ellipsoid-program-weights} if $w$ satisfies
\begin{align*}
&\sum_{i=1}^m w_i=n,\\
&a_i^{\top}\left(\sum_{i=1}^m w_ia_ia_i^{\top}\right)^{-1}a_i\leq 1+\epsilon, \qquad\forall i\in [m].
\end{align*}
\end{definition}
Lemma \ref{lem:approx_quality} gives a geometric interpretation of the approximation factor in Definition~\ref{def:weights-approx}. Recall that the exact John Ellipsoid $Q^{*}$ of $P$ satisfies $Q^{*}\subseteq P\subseteq \sqrt{n}\cdot Q^{*}$.
\begin{lemma}[$(1+\epsilon)$-approximation is good rounding]\label{lem:approx_quality}
Let $w$ be a $(1+\epsilon)$-approximation of \eqref{eq:ellipsoid-program-weights}.
Define $Q$ as $\{x:x^{\top}A^{\top}\diag(w)Ax\leq 1\}$.
Then
\[
\frac 1 {\sqrt{1+\epsilon}}\cdot Q\subseteq P\subseteq \sqrt{n}\cdot Q.
\]
Moreover,
$\vol\left(\frac 1 {\sqrt{1+\epsilon}}Q\right)\geq e^{-n\epsilon/2}\cdot \vol\left(Q^{*}\right)$.
\end{lemma}
\begin{proof}
Let $G=\left(A^{\top}\diag(w)A\right)^{-\frac 1 2}$ and suppose $x\in \frac 1 {\sqrt{1+\epsilon}}Q$. Then, we have that
$x^{\top}G^{-2}x\leq \frac{1}{1+\epsilon}$.
So,
\[
|Ax|_i=\langle a_i,x\rangle=\langle Ga_i,G^{-1}x\rangle\leq \|Ga_i\|_2\|G^{-1}x\|_2\leq \frac{\|Ga_i\|_2}{\sqrt{1+\epsilon}}.
\]
Since 
$\|Ga_i\|_2^2=a_i^{\top}(\sum_{i=1}^m w_ia_ia_i^{\top})^{-1}a_i\leq 1+\epsilon$,
then $|Ax|_i\leq 1$ and $x\in P$. 

On the other hand,
for $x\in P$,
we have that $|Ax|_i\leq 1$. 
Hence 
\[
x^{\top}G^{-2}x=x^{\top}A^{\top}\diag(w)Ax= \sum_{i=1}^m w_i |Ax|_i^2\leq \sum_{i=1}^m w_i= n.
\]
So $P\subseteq \sqrt{n}\cdot Q$.

Finally,
since $\frac 1 {\sqrt{1+\epsilon}}\cdot Q$ is contained in $P$,
$G'=((1+\epsilon)A^{\top}\diag(w)A)^{-\frac 1 2}$ is a feasible solution to program \eqref{eq:ellipsoid-program-original}.
Moreover $w$ is a feasible solution to program \eqref{eq:ellipsoid-program-weights}.
So by duality of program \eqref{eq:ellipsoid-program-original} and \eqref{eq:ellipsoid-program-weights},
we have the duality gap is at most
\[
\left(n-\log \det\left(\sum_{i=1}^m w_ia_ia_i^{\top}\right)-n\right)-\log \det\left((1+\epsilon)\sum_{i=1}^m w_ia_ia_i^{\top}\right)^{-1}=n\log (1+\epsilon)\leq n\epsilon.
\]
Let the matrix representation of $Q^{*}$ be $x^\top G_*^{-2}x\leq 1$,
then by optimality of $G_*$ and the duality gap, we have
\[
\log \det\left((G')^{-2}\right)\geq \log \det(G_*^2)-n\epsilon.
\]
Since $\vol(\frac 1 {\sqrt{1+\epsilon}}\cdot Q)$ is proportional to $\det(G')^{-1}$,
we conclude that $\vol(\frac 1 {\sqrt{1+\epsilon}}Q)\geq e^{-n\epsilon/2}\vol(Q^{*})$.
\end{proof}

\section{Main Algorithm}
In this section, we present Algorithm \ref{alg:short-faster} for approximating program \eqref{eq:ellipsoid-program-weights} and analyze its performance.

\begin{algorithm}[!t]
\LinesNumbered
\KwIn{A symmetric polytope given by $-\mathbf{1}_m\leq Ax\leq \mathbf{1}_m$, where $A\in \mathbb{R}^{m\times n}$ has rank $n$.}
\KwResult{Approximate John Ellipsoid inside the polytope.}
Initialize $w_{i}^{(1)}=\frac n m$ for $i=1,\cdots,m$.\\
\For{$k=1,\cdots,T-1$,}{
 	\For{$i=1,\cdots,m$}{
 	\tcp{We can use sketch technique to further speed up.}
	 $w_{i}^{(k+1)}=w_i^{(k)}\cdot a_i^{\top}(A^{\top}\diag(w^{(k)})A)^{-1}a_i$
	}

}
$w_i=\frac 1 T\sum_{k=1}^T w_i^{(k)}$ for $i=1,\cdots,m$. \label{alg:output}\\
$W=\diag(w)$. (i.e. $W$ is a diagonal matrix with the entries of $w$)

\Return $A^{\top}WA$
\caption{Approximate John Ellipsoid inside symmetric polytopes}\label{alg:short-faster}
\end{algorithm}

Let $\sigma:\mathbb{R}^m\rightarrow \mathbb{R}^m$ be the function defined as $\sigma(v)=(\sigma_1(v),\sigma_2(v),\cdots,\sigma_m(v))$ where for $i\in [m]$,
\begin{equation}\label{eq:short-sigma-def}
\sigma_i(v)=a_i^{\top}\left(\sum_{j=1}^m v_ja_ja_j^{\top}\right)^{-1}a_i=a_i^{\top}(A^{\top}\diag(v)A)^{-1}a_i.
\end{equation}
Let $w^{*}$ be the optimal solution to program \eqref{eq:ellipsoid-program-weights}.
By Lemma \eqref{lem:opt-criteria},
$w^{*}$ satisfies $w^{*}_i(1-\sigma_i(w^{*}))=0$,
or equivalently
\begin{align}\label{eq:kkt}
w^{*}_i=w^{*}_i\cdot \sigma_i(w^{*})
\end{align}
Inspired by \eqref{eq:kkt},
we use the fixed point iteration $w^{(k+1)}_i=w^{(k)}_i\cdot \sigma_i(w^{(k)})$ for $k\in [T-1]$ and  $i\in [m]$.
$w_{i}^{(k)}$ has very nice properties.
Actually,
by setting $B^{(k)}=\sqrt{\diag(w^{(k)})}\cdot A$,
we can rewrite $w_{i}^{(k)}$ as $(B^{(k)}_i)^{\top}\left((B^{(k)})^{\top}B^{(k)}\right)^{-1}B^{(k)}_i$,
hence $w_{i}^{(k)}$ is actually the \emph{leverage score} of the $i$-th row of the matrix $B^{(k)}$~\cite{clmmps15}.
From the well-known properties of leverage scores, we have 
\begin{lemma}[Properties of leverage scores, e.g. see Section 3.3 of \cite{clmmps15}]\label{lem:short-w_property}
For $k\in [T]$ and $i\in [m]$, we have
$0\leq  w_{i}^{(k)}\leq 1$.
Moreover,
$\sum_{i=1}^m w_{i}^{(k)}=n$.
\end{lemma}

In order to show Algorithm \ref{alg:short-faster} provides a good approximation of the John Ellipsoid,
in the sense of Definition \ref{def:weights-approx},
we need to argue that for the output $w$ of Algorithm \ref{alg:short-faster},
$\sigma_i(w)\leq 1+\epsilon$.
Our main result is the following theorem.

\begin{theorem}[Main Result]\label{thm:short-good-approx}
Let $w$ be the output of Algorithm \ref{alg:short-faster} in line \eqref{alg:output}.
For all $\epsilon\in (0,1)$,
when $T=\frac {2}{\epsilon}\log \frac m {n}$,
we have for $i\in [m]$,
\[
\sigma_i(w)\leq 1+\epsilon.
\]
Moreover,
\[
\sum_{i=1}^m w_i=n.
\]
Therefore,
Algorithm \ref{alg:short-faster} provides  $(1+\epsilon)$-approximation to program \eqref{eq:ellipsoid-program-weights}.
\end{theorem}

We now analyze the running time of Algorithm \ref{alg:short-faster}.

\begin{theorem}[Performance of Algorithm \ref{alg:short-faster}]
For all $\epsilon\in (0,1)$,
we can find a $(1+\epsilon)$-approximation of John Ellipsoid inside a symmetric convex polytope in time $O\left(\epsilon^{-1}mn^2\log \frac m n\right)$.
\end{theorem}
\begin{proof}	
	The main loop is executed $T=O(\frac 1 \eps \log (\frac m n))$ times,
	and inside each loop,
	we can first use $O(mn)$ time to compute $B^{(k)}:=(W^{(k)})^{\frac 1 2}A$,
	then compute $(B^{(k)})^{\top}B^{(k)}$ in $O(mn^2)$ time.	
	To see why we introduce $B^{(k)}$,
	observe that $(B^{(k)})^{\top}B^{(k)}=A^{\top}W^{(k)}A$.
	Now we can compute the Cholesky decomposition of $(B^{(k)})^{\top}B^{(k)}$ in time $O(n^3)$,
	and use the Cholesky decomposition to compute $c_i:=((B^{(k)})^{\top}B^{(k)})^{-1}a_i=(A^{\top}W^{(k)}A)^{-1}a_i$
	in time $O(n^2)$ for each $i\in [m]$.
	Finally, we can compute $w^{(k+1)}_i$ by computing $w_i^{(k)}\cdot a_i^\top c_i$ in time $O(n)$.
	This is valid since $w_i^{(k)}\cdot a_i^\top c_i=w_i^{(k)}\cdot a_i^\top(A^{\top}W^{(k)}A)^{-1}a_i=w_i^{(k)}\sigma_i(w^{(k)})$.
	To summarize,
	in each iteration we use $O(mn^2+n^3+mn^2)=O(mn^2)$ time,
	hence the overall running time is as stated.
\end{proof}

Now we turn to proving Theorem \ref{thm:short-good-approx}.
The proof of Theorem \ref{thm:short-good-approx} relies on the following important observation,
whose proof can be found in Appendix \ref{sec:proof-convexity}.
\begin{lemma}[Convexity]\label{lem:convexity}
For $i=1,\cdots,m$,
let $\phi_i:\mathbb{R}^{m}\rightarrow \mathbb{R}$ be the function defined as
\[
\phi_i(v)=\log \sigma_i(v)=\log \left(a_i^{\top}\left(\sum_{j=1}^m v_ja_ja_j^{\top}\right)^{-1}a_i\right).
\]
Then $\phi_i$ is convex.

\end{lemma}
Now that $\phi_i$ is convex,
we can apply Jensen's inequality to get Lemma \ref{lem:jenson}.
\begin{lemma}[Telescoping]\label{lem:jenson}
Fix $T$ as the number of main loops executed in Algorithm \ref{alg:short-faster}.
Let $w$ be the output in line \eqref{alg:output} of Algorithm \ref{alg:short-faster}.
Then for $i\in [m]$,
\[
\phi_i(w)\leq \frac 1 T\log \frac m n
\]
\end{lemma}
\begin{proof}
Recall that $w=\frac 1 T\sum_{k=1}^T w^{(k)}$.
By Lemma \ref{lem:convexity},
$\phi_i$ is convex, and so
\begin{align*}
\phi_i(w)&=\phi_i\left(\frac 1 T\sum_{k=1}^Tw^{(k)}\right)\\
&\leq \frac 1 T\sum_{k=1}^T\phi_i(w^{(k)})  &\text{by Jensen's inequality}\\
&=\frac 1 T\sum_{k=1}^T\log \sigma_i(w^{(k)}) &\text{by definition of $\phi_i$ function } \\
&= \frac 1 T\sum_{k=1}^{T}\log \frac{w^{(k+1)}_i}{w^{(k)}_i}  &\\
&=\frac 1 T\log \frac{ w^{(T+1)}_i}{w^{(1)}_i}&\\
&\leq \frac 1 T\log \frac m n &\text{by Lemma \ref{lem:short-w_property} and the initialization of $w^{(1)}$}
\end{align*}
\end{proof}

Now we are ready to prove Theorem \ref{thm:short-good-approx}.
\begin{proof}[Proof of Theorem \ref{thm:short-good-approx}]
Set $T=\frac 2 {\epsilon}\log \frac m n$.
By Lemma \ref{lem:jenson},
we have for $i\in [m] $,
\[
\log \sigma_i(w)=\phi_i(w)\leq \frac 1 T\log \frac m n=\frac {\epsilon} 2\leq \log(1+\epsilon)
\]
where the last step uses the fact that when $0<\epsilon<1$,
$\frac {\epsilon} 2\leq \log(1+\epsilon)$.
This gives us $\sigma_i(w)\leq 1+\epsilon$.

On the other hand,
from Lemma \ref{lem:short-w_property} we have $\sum_{i=1}^m w^{(k)}_i=n$.
Hence 
\[
\sum_{i=1}^m w_i=\sum_{i=1}^m\frac 1 T\sum_{k=1}^Tw^{(k)}_i=n
\]
\end{proof}

We shall mention that it is possible to further improve Algorithm \ref{alg:short-faster} by applying sketching technique from randomized linear algebra.
Here we present the performance of our accelerated algorithm, and detailed analysis can be found in Appendix \ref{sec:algorithm-with-JL}. 
\begin{theorem}[Performance of Algorithm \ref{alg:faster}]\label{thm:short-performance-jl}
For all $\epsilon,\delta\in (0,1)$,
we can find a $(1+\epsilon)$-approximation of the John Ellipsoid inside a symmetric convex polytope within $O\Big(\frac {1} {\epsilon}\log \frac m {\delta}\Big)$ iterations with probability at least $1-\delta$.
Moreover,
each iteration involves solving $O(\frac 1 \epsilon)$ linear systems of the form $A^{\top}WAx=b$ for some diagonal matrix $W$.
\end{theorem}

\section*{Acknowledgment}
We thank Zhao Song for his generous help on this paper.

\bibliographystyle{alpha}
\bibliography{mybib}

\newcommand{\etalchar}[1]{$^{#1}$}
\begin{thebibliography}{DMIMW12}

\bibitem[ALS{\etalchar{+}}18]{alszz18}
Alexandr Andoni, Chengyu Lin, Ying Sheng, Peilin Zhong, and Ruiqi Zhong.
\newblock Subspace embedding and linear regression with {O}rlicz norm.
\newblock In {\em International Conference on Machine Learning}, pages
  224--233. \url{https://arxiv.org/pdf/1806.06430}, 2018.

\bibitem[Ans02]{ans02}
Kurt~M Anstreicher.
\newblock Improved complexity for maximum volume inscribed ellipsoids.
\newblock {\em SIAM Journal on Optimization}, 13(2):309--320, 2002.

\bibitem[Atw69]{atwood1969optimal}
Corwin~L Atwood.
\newblock Optimal and efficient designs of experiments.
\newblock {\em The Annals of Mathematical Statistics}, pages 1570--1602, 1969.

\bibitem[AZLSW17]{azlsw17}
Zeyuan Allen-Zhu, Yuanzhi Li, Aarti Singh, and Yining Wang.
\newblock Near-optimal design of experiments via regret minimization.
\newblock In {\em International Conference on Machine Learning}, pages
  126--135. \url{https://arxiv.org/pdf/1711.05174}, 2017.

\bibitem[BCBK12]{bck12}
S{\'e}bastien Bubeck, Nicolo Cesa-Bianchi, and Sham Kakade.
\newblock Towards minimax policies for online linear optimization with bandit
  feedback.
\newblock In {\em Annual Conference on Learning Theory}, volume~23, pages
  41--1. Microtome, 2012.

\bibitem[BV04]{BoydVandenbergheBook}
Stephen Boyd and Lieven Vandenberghe.
\newblock {\em Convex optimization}.
\newblock Cambridge University Press, Cambridge, 2004.

\bibitem[BW14]{bw14}
Christos Boutsidis and David~P Woodruff.
\newblock Optimal {C}{U}{R} matrix decompositions.
\newblock In {\em Proceedings of the 46th Annual ACM Symposium on Theory of
  Computing (STOC)}, pages 353--362. ACM,
  \url{https://arxiv.org/pdf/1405.7910}, 2014.

\bibitem[CDWY18]{chen2018fast}
Yuansi Chen, Raaz Dwivedi, Martin~J Wainwright, and Bin Yu.
\newblock Fast {M}{C}{M}{C} sampling algorithms on polytopes.
\newblock {\em The Journal of Machine Learning Research}, 19(1):2146--2231,
  2018.

\bibitem[CEM{\etalchar{+}}15]{cemmp15}
Michael~B Cohen, Sam Elder, Cameron Musco, Christopher Musco, and Madalina
  Persu.
\newblock Dimensionality reduction for k-means clustering and low rank
  approximation.
\newblock In {\em Proceedings of the forty-seventh annual ACM symposium on
  Theory of computing}, pages 163--172. ACM, 2015.

\bibitem[{\v{C}}H12]{vh12}
Michal {\v{C}}ern{\`y} and Milan Hlad{\'\i}k.
\newblock Two complexity results on {C}-optimality in experimental design.
\newblock {\em Computational Optimization and Applications}, 51(3):1397--1408,
  2012.

\bibitem[CKM{\etalchar{+}}14]{ckmpprx14}
Michael~B. Cohen, Rasmus Kyng, Gary~L. Miller, Jakub~W. Pachocki, Richard Peng,
  Anup~B. Rao, and Shen~Chen Xu.
\newblock Solving {S}{D}{D} linear systems in nearly $m$log$^{1/2}n$ time.
\newblock In {\em Proceedings of the Forty-sixth Annual ACM Symposium on Theory
  of Computing}, STOC '14, 2014.

\bibitem[CLM{\etalchar{+}}15]{clmmps15}
Michael~B Cohen, Yin~Tat Lee, Cameron Musco, Christopher Musco, Richard Peng,
  and Aaron Sidford.
\newblock Uniform sampling for matrix approximation.
\newblock In {\em Proceedings of the 2015 Conference on Innovations in
  Theoretical Computer Science}, pages 181--190. ACM, 2015.

\bibitem[CP15]{cp15}
Michael~B Cohen and Richard Peng.
\newblock $\ell_p$ row sampling by {L}ewis {W}eights.
\newblock In {\em Proceedings of the forty-seventh annual ACM symposium on
  Theory of computing}, pages 183--192. ACM,
  \url{https://arxiv.org/pdf/1412.0588}, 2015.

\bibitem[CW13]{cw13}
Kenneth~L Clarkson and David~P Woodruff.
\newblock Low rank approximation and regression in input sparsity time.
\newblock In {\em Proceedings of the forty-fifth annual ACM symposium on Theory
  of computing}, pages 81--90. ACM, 2013.

\bibitem[CW15a]{cw15b}
Kenneth~L Clarkson and David~P Woodruff.
\newblock Input sparsity and hardness for robust subspace approximation.
\newblock In {\em 2015 IEEE 56th Annual Symposium on Foundations of Computer
  Science}, pages 310--329. IEEE, 2015.

\bibitem[CW15b]{cw15a}
Kenneth~L Clarkson and David~P Woodruff.
\newblock Sketching for {M}-estimators: A unified approach to robust
  regression.
\newblock In {\em Proceedings of the twenty-sixth annual ACM-SIAM symposium on
  Discrete algorithms}, pages 921--939. Society for Industrial and Applied
  Mathematics, 2015.

\bibitem[DAST08]{dst08}
S~Damla~Ahipasaoglu, Peng Sun, and Michael~J Todd.
\newblock Linear convergence of a modified {F}rank--{W}olfe algorithm for
  computing minimum-volume enclosing ellipsoids.
\newblock {\em Optimisation Methods and Software}, 23(1):5--19, 2008.

\bibitem[DMIMW12]{dmmw12}
Petros Drineas, Malik Magdon-Ismail, Michael~W Mahoney, and David~P Woodruff.
\newblock Fast approximation of matrix coherence and statistical leverage.
\newblock {\em Journal of Machine Learning Research}, 13(Dec):3475--3506, 2012.

\bibitem[DS08]{ds08}
Samuel~I Daitch and Daniel~A Spielman.
\newblock Faster approximate lossy generalized flow via interior point
  algorithms.
\newblock In {\em Proceedings of the fortieth annual ACM symposium on Theory of
  computing}, pages 451--460. ACM, 2008.

\bibitem[DSSW18]{dssw18}
Huaian Diao, Zhao Song, Wen Sun, and David Woodruff.
\newblock Sketching for kronecker product regression and p-splines.
\newblock In {\em International Conference on Artificial Intelligence and
  Statistics}, pages 1299--1308, 2018.

\bibitem[DSWY19]{dswy19}
Huaian Diao, Zhao Song, David Woodruff, and Xin Yang.
\newblock Total least squares regression in input sparsity time.
\newblock In {\em Manuscript}, 2019.

\bibitem[GN18]{gustafson2018john}
Adam Gustafson and Hariharan Narayanan.
\newblock John's walk.
\newblock {\em arXiv preprint arXiv:1803.02032}, 2018.

\bibitem[GP18]{gp18}
David~H Gutman and Javier~F Pe{\~n}a.
\newblock A unified framework for bregman proximal methods: subgradient,
  gradient, and accelerated gradient schemes.
\newblock {\em arXiv preprint arXiv:1812.10198}, 2018.

\bibitem[GR14]{gradshteyn2014table}
Izrail~Solomonovich Gradshteyn and Iosif~Moiseevich Ryzhik.
\newblock {\em Table of integrals, series, and products}.
\newblock Academic press, 2014.

\bibitem[HCT{\etalchar{+}}17]{chrr2017}
Hulda~S. Haraldsdottir, Ben Cousins, Ines Thiele, Ronan M.~T. Fleming, and
  Santosh Vempala.
\newblock {C}{H}{R}{R}: coordinate hit-and-run with rounding for uniform
  sampling of constraint-based models.
\newblock {\em Bioinformatics}, 33(11), 1 2017.

\bibitem[HK16]{hazan2016volumetric}
Elad Hazan and Zohar Karnin.
\newblock Volumetric spanners: an efficient exploration basis for learning.
\newblock {\em The Journal of Machine Learning Research}, 17(1):4062--4095,
  2016.

\bibitem[Jam13]{jameson2013inequalities}
GJO Jameson.
\newblock Inequalities for gamma function ratios.
\newblock {\em The American Mathematical Monthly}, 120(10):936--940, 2013.

\bibitem[Joh48]{John48}
Fritz John.
\newblock Extremum problems with inequalities as subsidiary conditions.
\newblock In {\em Studies and {E}ssays {P}resented to {R}. {C}ourant on his
  60th {B}irthday, {J}anuary 8, 1948}, pages 187--204. Interscience Publishers,
  Inc., New York, N. Y., 1948.

\bibitem[Kha96]{kha96}
Leonid~G Khachiyan.
\newblock Rounding of polytopes in the real number model of computation.
\newblock {\em Mathematics of Operations Research}, 21(2):307--320, 1996.

\bibitem[KLP{\etalchar{+}}16]{klpss16}
Rasmus Kyng, Yin~Tat Lee, Richard Peng, Sushant Sachdeva, and Daniel~A
  Spielman.
\newblock Sparsified cholesky and multigrid solvers for connection laplacians.
\newblock In {\em Proceedings of the forty-eighth annual ACM symposium on
  Theory of Computing}, pages 842--850. ACM, 2016.

\bibitem[KLS97]{KannanLovaszSimonovits97}
Ravi Kannan, L\'{a}szl\'{o} Lov\'{a}sz, and Mikl\'{o}s Simonovits.
\newblock Random walks and an {$O^*(n^5)$} volume algorithm for convex bodies.
\newblock {\em Random Structures Algorithms}, 11(1):1--50, 1997.

\bibitem[KOSZ13]{kosz13}
Jonathan~A Kelner, Lorenzo Orecchia, Aaron Sidford, and Zeyuan~Allen Zhu.
\newblock A simple, combinatorial algorithm for solving {S}{D}{D} systems in
  nearly-linear time.
\newblock In {\em Proceedings of the forty-fifth annual ACM symposium on Theory
  of computing}, pages 911--920. ACM, 2013.

\bibitem[KS16]{ks16}
Rasmus Kyng and Sushant Sachdeva.
\newblock Approximate gaussian elimination for laplacians-fast, sparse, and
  simple.
\newblock In {\em 2016 IEEE 57th Annual Symposium on Foundations of Computer
  Science (FOCS)}, pages 573--582. IEEE, 2016.

\bibitem[KT93]{kt93}
Leonid~G Khachiyan and Michael~J Todd.
\newblock On the complexity of approximating the maximal inscribed ellipsoid
  for a polytope.
\newblock {\em Mathematical Programming}, 61(1):137--159, 1993.

\bibitem[KTE88]{Khachiyan88}
L.~Khachiyan, S.~Tarasov, and I.~Ehrlich.
\newblock The method of inscribed ellipsoids.
\newblock {\em Soviet Math. Doklady}, 1988.

\bibitem[KW60]{kiefer1960equivalence}
Jack Kiefer and Jacob Wolfowitz.
\newblock The equivalence of two extremum problems.
\newblock {\em Canadian Journal of Mathematics}, 12(363-366):234, 1960.

\bibitem[KY05]{ky05}
Piyush Kumar and E~Alper Yildirim.
\newblock Minimum-volume enclosing ellipsoids and core sets.
\newblock {\em Journal of Optimization Theory and Applications}, 126(1):1--21,
  2005.

\bibitem[LFN18]{lfn18}
Haihao Lu, Robert~M Freund, and Yurii Nesterov.
\newblock Relatively smooth convex optimization by first-order methods, and
  applications.
\newblock {\em SIAM Journal on Optimization}, 28(1):333--354, 2018.

\bibitem[LHW17]{lhw17}
Xingguo Li, Jarvis Haupt, and David Woodruff.
\newblock Near optimal sketching of low-rank tensor regression.
\newblock In {\em Advances in Neural Information Processing Systems}, pages
  3466--3476. \url{https://arxiv.org/pdf/1709.07093}, 2017.

\bibitem[LM00]{laurent2000adaptive}
Beatrice Laurent and Pascal Massart.
\newblock Adaptive estimation of a quadratic functional by model selection.
\newblock {\em Annals of Statistics}, pages 1302--1338, 2000.

\bibitem[LS14]{ls14}
Yin~Tat Lee and Aaron Sidford.
\newblock Path finding methods for linear programming: Solving linear programs
  in ${O} (\sqrt{rank})$ iterations and faster algorithms for maximum flow.
\newblock In {\em Foundations of Computer Science (FOCS), 2014 IEEE 55th Annual
  Symposium on}, pages 424--433. IEEE, 2014.

\bibitem[LSZ19]{lsz19}
Yin~Tat Lee, Zhao Song, and Qiuyi Zhang.
\newblock Solving empirical risk minimization in the current matrix
  multiplication time.
\newblock In {\em COLT}. \url{https://arxiv.org/pdf/1905.04447}, 2019.

\bibitem[LV06]{LovaszVempala06}
L\'{a}szl\'{o} Lov\'{a}sz and Santosh Vempala.
\newblock Hit-and-run from a corner.
\newblock {\em SIAM J. Comput.}, 35(4):985--1005, 2006.

\bibitem[Nem99]{ne99}
Arkadi Nemirovski.
\newblock On self-concordant convex--concave functions.
\newblock {\em Optimization Methods and Software}, 11(1-4):303--384, 1999.

\bibitem[NN94]{nn94}
Yurii Nesterov and Arkadii Nemirovskii.
\newblock {\em Interior-point polynomial algorithms in convex programming},
  volume~13.
\newblock Siam, 1994.

\bibitem[NN13]{nn13}
Jelani Nelson and Huy~L Nguy{\^e}n.
\newblock {O}{S}{N}{A}{P}: Faster numerical linear algebra algorithms via
  sparser subspace embeddings.
\newblock In {\em Foundations of Computer Science (FOCS), 2013 IEEE 54th Annual
  Symposium on}, pages 117--126. IEEE, 2013.

\bibitem[NTZ13]{ntz13}
Aleksandar Nikolov, Kunal Talwar, and Li~Zhang.
\newblock The geometry of differential privacy: the sparse and approximate
  cases.
\newblock In {\em Proceedings of the forty-fifth annual ACM symposium on Theory
  of computing}, pages 351--360. ACM, 2013.

\bibitem[PSW17]{psw17}
Eric Price, Zhao Song, and David~P. Woodruff.
\newblock Fast regression with an ${\ell}_{\infty}$ guarantee.
\newblock In {\em International Colloquium on Automata, Languages, and
  Programming (ICALP)}, 2017.

\bibitem[RSW16]{rsw16}
Ilya Razenshteyn, Zhao Song, and David~P Woodruff.
\newblock Weighted low rank approximations with provable guarantees.
\newblock In {\em Proceedings of the forty-eighth annual ACM symposium on
  Theory of Computing}, pages 250--263. ACM, 2016.

\bibitem[SF04]{sf04}
Peng Sun and Robert~M Freund.
\newblock Computation of minimum-volume covering ellipsoids.
\newblock {\em Operations Research}, 52(5):690--706, 2004.

\bibitem[SS11]{ss11}
Daniel~A Spielman and Nikhil Srivastava.
\newblock Graph sparsification by effective resistances.
\newblock {\em SIAM Journal on Computing}, 40(6):1913--1926, 2011.

\bibitem[SWZ16]{swz16}
Zhao Song, David~P. Woodruff, and Huan Zhang.
\newblock Sublinear time orthogonal tensor decomposition.
\newblock In {\em Advances in Neural Information Processing Systems 29: Annual
  Conference on Neural Information Processing Systems (NIPS) 2016, December
  5-10, 2016, Barcelona, Spain}, pages 793--801, 2016.

\bibitem[SWZ17]{swz17}
Zhao Song, David~P Woodruff, and Peilin Zhong.
\newblock Low rank approximation with entrywise $\ell_1$-norm error.
\newblock In {\em Proceedings of the 49th Annual ACM SIGACT Symposium on Theory
  of Computing}, pages 688--701. ACM, 2017.

\bibitem[SWZ19]{swz19}
Zhao Song, David~P Woodruff, and Peilin Zhong.
\newblock Relative error tensor low rank approximation.
\newblock In {\em SODA}. \url{https://arxiv.org/pdf/1704.08246}, 2019.

\bibitem[SX18]{sx18}
Mohit Singh and Weijun Xie.
\newblock Approximate positive correlated distributions and approximation
  algorithms for {D}-optimal design.
\newblock In {\em Proceedings of the Twenty-Ninth Annual ACM-SIAM Symposium on
  Discrete Algorithms}, pages 2240--2255. Society for Industrial and Applied
  Mathematics, 2018.

\bibitem[SYZ18]{syz18}
Zhao Song, Lin~F Yang, and Peilin Zhong.
\newblock Sensitivity sampling over dynamic geometric data streams with
  applications to $ k $-clustering.
\newblock {\em arXiv preprint arXiv:1802.00459}, 2018.

\bibitem[Tod16]{Todd2016}
Michael~J. Todd.
\newblock {\em Minimum-volume ellipsoids : theory and algorithms}.
\newblock MOS-SIAM series on optimization. SIAM, Philadelphia, 2016.

\bibitem[TY07]{ty07}
Michael~J Todd and E~Alper Y{\i}ld{\i}r{\i}m.
\newblock On khachiyan's algorithm for the computation of minimum-volume
  enclosing ellipsoids.
\newblock {\em Discrete Applied Mathematics}, 155(13):1731--1744, 2007.

\bibitem[Vem05]{VempalaSurvey}
Santosh Vempala.
\newblock Geometric random walks: a survey.
\newblock In {\em Combinatorial and computational geometry}, volume~52 of {\em
  Math. Sci. Res. Inst. Publ.}, pages 577--616. Cambridge Univ. Press,
  Cambridge, 2005.

\bibitem[Woo14]{w14}
David~P Woodruff.
\newblock Sketching as a tool for numerical linear algebra.
\newblock {\em Foundations and Trends{\textregistered} in Theoretical Computer
  Science}, 10(1--2):1--157, 2014.

\bibitem[WTSA15]{wtsa15}
Yining Wang, Hsiao-Yu Tung, Alexander~J Smola, and Anima Anandkumar.
\newblock Fast and guaranteed tensor decomposition via sketching.
\newblock In {\em Advances in Neural Information Processing Systems (NIPS)},
  pages 991--999. \url{https://arxiv.org/pdf/1506.04448}, 2015.

\bibitem[WYS17]{wang2017computationally}
Yining Wang, Adams~Wei Yu, and Aarti Singh.
\newblock On computationally tractable selection of experiments in
  measurement-constrained regression models.
\newblock {\em The Journal of Machine Learning Research}, 18(1):5238--5278,
  2017.

\end{thebibliography}

\appendix
\section{Preliminaries}
In this section we introduce notations and preliminaries used in the appendix.
We use $N(\mu,\sigma^2)$ to represent the normal distribution with mean $\mu$ and variance $\sigma^2$.
\subsection{Multivariate Calculus}
Let $f:\R^{m}\rightarrow \R^{n}$ be a differentiable function.
The \emph{directional derivative} of $f$ in the direction $h$ is defined as
\[
D f(x)[h]:=\frac{\de f(x+th)}{\de t}\Big|_{t=0}.
\]
We can also define a high order directional derivative as
\[
D^k f(x)[h_1,\cdots,h_k]:=\frac{\de^k f(x+\sum_{i=1}^k t_ih_i)}{\de t_1\de t_2\cdots \de t_k}\Big|_{t_1=0, \ldots, t_k = 0}.
\]

The following two properties of directional derivatives will be useful.

\begin{proposition}\label{prop:directional-derivative}
\begin{itemize}
\item (Chain rule) $D f(g(x))[h]=f'(g(x))\cdot D g(x)[h]$.
\item Let $f(X)=X^{-1}$ where $X\in \R^{n\times n}$. For $H\in \R^{n\times n}$, $D f(X)[H]=X^{-1}HX^{-1}$.
\end{itemize}
\end{proposition}




\subsection{Gamma Function}
$\Gamma$ function is a well-known math object. It is defined as
\[
\Gamma(z)=\int_0^{+\infty}x^{z-1}e^{-x}\de x.
\] 
We need the following result on Gamma function.
\begin{lemma}[Corollary 1 of \cite{jameson2013inequalities}]\label{lem:gamma-function}
For all $x>0$ and $0\leq y\leq 1$,
\[
x(x+y)^{y-1}\leq \frac {\Gamma(x+y)}{\Gamma(x)}\leq x^y.
\]
\end{lemma}

\subsection{Tail Bound for $\chi^2$ Distribution}

We need the following version of concentration for $\chi^2$ distribution.
\begin{lemma}[Lemma 1 in \cite{laurent2000adaptive}]\label{lem:chi-square-tail}
Let $X\sim \chi^2(n)$ be a $\chi^2$ distribution with $n$ degree of freedom.
Then for $t>0$,
\[
\Pr[X-n\geq 2\sqrt{nt}+2t]\leq e^{-t}.
\]
\end{lemma}
\section{Proof of Lemma \ref{lem:convexity}}\label{sec:proof-convexity}
In this section we provide the proof of Lemma \ref{lem:convexity}.
\begin{proof}
We first prove a strengthened result: for fixed $a\in \mathbb{R}^{n}$,
the function $f:S^n_{++}\rightarrow \mathbb{R}$ defined as $f(M)=\log (a^{\top}M^{-1}a)$ is convex.
Here $S^n_{++}$ is the set of all positive definite $n\times n$ matrices.
Notice that $S^n_{++}$ is an open set,
so we can differentiate $f$.

We argue that it is sufficient to show for all $M\in GL_n$ and all $H\in \mathbb{R}^{n\times n}$,
the second order directional derivative $D^2 f(M)[H,H]$ is non-negative.
This is because $D^2 f(M)[H,H]=H^{\top}\nabla^2f(M)H$. 
So if for all $H$ we have $H^{\top}\nabla^2f(M)H\geq 0$,
then $\nabla^2f(M) \succeq 0$,
which is precisely the convex condition.

Let us do some computation with Proposition \ref{prop:directional-derivative}.
\[
Df(M)[H]=-\frac{a^{\top}M^{-1}HM^{-1}a}{a^{\top}M^{-1}a},
\]
and
\[
D^2f(M)[H,H]=\frac{2a^{\top}M^{-1}HM^{-1}HM^{-1}a\cdot a^{\top}M^{-1}a-(a^{\top}M^{-1}HM^{-1}a)^2}{(a^{\top}M^{-1}a)^2}.
\]
By Cauchy-Schwarzt inequality,
we have
\[
\begin{split}
a^{\top}M^{-1}HM^{-1}HM^{-1}a\cdot a^{\top}M^{-1}a
=&\|M^{-\frac 1 2}HM^{-1}a\|_2^2\cdot \|M^{-\frac 1 2}a\|_2^2\\
\geq& (\langle M^{-\frac 1 2}HM^{-1}a,M^{-\frac 1 2}a\rangle)^2\\
=&(a^{\top}M^{-1}HM^{-1}a)^2.
\end{split}
\]
Hence $D^2f(M)[H,H]\geq \frac{a^{\top}M^{-1}HM^{-1}HM^{-1}a\cdot a^{\top}M^{-1}a}{(a^{\top}M^{-1}a)^2}\geq 0$ for all $M\in GL_n$ and all $H\in \mathbb{R}^{n\times n}$.

Now we are ready to work on $\phi_i$.
For all $v,v'$ in the domain of $\phi_i$,
let $M=\sum_{i=1}^m v_ia_ia_i^{\top}$ and $M'=\sum_{i=1}^m v_i'a_ia_i^{\top}$.
Then for all $\lambda \in [0,1]$,
\[
\begin{split}
\phi_i(\lambda v+(1-\lambda)v')&=\log a_i^{\top}\left(\sum_{i=1}^m (\lambda v_i+(1-\lambda)v_i')a_ia_i^{\top}\right)^{-1}a_i\\
&=\log a_i^{\top}\left(\lambda\sum_{i=1}^m v_ia_ia_i^{\top}+(1-\lambda)\sum_{i=1}^m v_i'a_ia_i^{\top}\right)^{-1}a_i\\
&=f(\lambda M+(1-\lambda)M')\\
&\leq \lambda f(M)+(1-\lambda)f(M') \qquad \text{because $f$ is convex}\\
&=\lambda\phi_i(v)+(1-\lambda)\phi_i(v').
\end{split}
\]
So $\phi_i$ is also convex.
\end{proof}

\section{Faster Algorithm for Computing John Ellipsoid for Sparse Matrix}\label{sec:algorithm-with-JL}

\begin{algorithm}[!t]
\LinesNumbered
\KwIn{A symmetric polytope given by $-\mathbf{1}_m\leq Ax\leq \mathbf{1}_m$, where $A\in \mathbb{R}^{m\times n}$}
\KwResult{Approximate John Ellipsoid inside the polytope}
initialize $w_{i}^{(1)}=\frac n m$ for $i=1,\cdots,m$.\\
\For{$k=1,\cdots,T-1$,}{
 	$W^{(k)}=\diag(w^{(k)})$.\\
	$B^{(k)}=\sqrt{W^{(k)}}A$.\\
 	Let $S^{(k)}\in \mathbb{R}^{s\times m}$ be a random matrix where each entry is chosen i.i.d from $N(0,1)$, i.e. the standard normal distribution.\\
 	\For{$i=1,\cdots,m$}{
 	\tcp{Ideally we want to compute $\hat w_{i}^{(k+1)}=\|B^{(k)}((B^{(k)})^{\top}B^{(k)})^{-1}(\sqrt{w_i^{(k)}}a_i)\|_2^2$.\\
	But this is expensive, so we use sketching technique to speed up.}
	 $w_{i}^{(k+1)}=\frac 1 s\cdot \|S^{(k)}B^{(k)}((B^{(k)})^{\top}B^{(k)})^{-1}(\sqrt{w_i^{(k)}}a_i)\|_2^2$.
	}

}
$w_i=\frac 1 T\sum_{k=1}^T w_i^{(k)}$ for $i=1,\cdots,m$.\label{alg:output-jl}\\
$v_i=\frac {n}{\sum_{j=1}^m w_j}w_i$ for $i=1,\cdots,m$.\label{alg:rescaling}\\
$V=\diag(v)$.

\Return $A^{\top}VA$
\caption{Faster Algorithm for approximating John Ellipsoid inside symmetric polytopes}\label{alg:faster}
\end{algorithm}

In this section we present our accelerated algorithm,
 Algorithm \ref{alg:faster} and analyze its performance.
Recall that Algorithm \ref{alg:short-faster} uses the iterating rule $w\leftarrow w\cdot\sigma(w)$ where 
$\sigma_i(w)=a_i^{\top}(A^{\top}\diag(w)A)^{-1}a_i$.
With our setting of $B^{(k)}$,
we have
\[
(B^{(k)})^{\top}B^{(k)}=A^{\top}\sqrt{W^{(k)}}\cdot \sqrt{W^{(k)}}A=A^{\top}W^{(k)}A.
\]
So $\hat w_{i}^{(k+1)}=\|B^{(k)}((B^{(k)})^{\top}B^{(k)})^{-1}(\sqrt{w_i^{(k)}}a_i)\|_2^2=w_i\cdot \sigma_i(w)$ is just what we do in Algorithm \ref{alg:short-faster}.
Hence from Lemma \ref{lem:short-w_property},
we obtain the following properties about $\hat w_{i}^{(k)}$.
\begin{proposition}[Bound on $\hat{w}^{(k)}$]\label{prop:hatw_property}
For completeness we define $\hat w^{(1)}=w^{(1)}$.
For $k\in [T]$ and $i\in [m]$,
$0\leq  \hat w_{i}^{(k)}\leq 1$.
Moreover,
$\sum_{i=1}^m \hat w_{i}^{(k)}=n$.
\end{proposition}

However,
$B^{(k)}$ is a $m$ by $n$ matrix,
so it is computationally expensive to compute $\hat w_{i}^{(k+1)}$.
The trick we use here,
which is initially introduced first by \cite{ss11},
is to introduce a random Gaussian matrix $S$ with $s$ rows to speed up the computation.
Of course,
this will introduce extra error,
however we can prove that the overall result has good concentration.

\subsection{Approximation Guarantee}

Since Algorithm \ref{alg:faster} is a randomized algorithm,
we need to argue that for the output $v$ of Algorithm \ref{alg:faster},
$\sigma_i(v)\leq 1+\epsilon$ with high probability.
Our main result in this section is 

\begin{theorem}[Main result]\label{thm:good-approx}
Let $w$ be the output in line \eqref{alg:output-jl} of Algorithm \ref{alg:faster}.
For all $\epsilon,\delta\in (0,1)$,
when $T=\frac {10}{\epsilon}\log \frac m {\delta}$ and $s=\frac {80}{\epsilon}$,
we have
\[
\Pr[\forall i\in [m], \sigma_i(w)\leq 1+\epsilon]\geq 1-\delta.
\]
Moreover, before rescaling at the end,
\[
\Pr\left[\sum_{i=1}^m w_i\leq (1+\epsilon)n\right]\geq 1-\delta.
\]
\end{theorem}
By scaling $w$ so that $\sum_{i=1}^m w_i=n$,
we have
\begin{theorem}[Approximation guarantee]\label{thm:algorithm-performance}
For all $\epsilon,\delta\in (0,1)$,
When $T=O(\frac 1 {\epsilon}\log \frac m {\delta})$ and $s=O(\frac 1 {\epsilon})$,
Algorithm \ref{alg:faster} provides a $(1+\epsilon)^2$-approximation to program \eqref{eq:ellipsoid-program-weights} with probability at least $1-2\delta$.
\end{theorem}
\begin{proof}
On line \eqref{alg:rescaling} of Algorithm \ref{alg:faster} we set $v=\frac n {\sum_{i=1}^m w_i}w$,
hence $\sum_{i=1}^m v_i=n$.

From Theorem \ref{thm:good-approx},
with probability at least $1-2\delta$,
$\forall i\in [m], \sigma_i(w)\leq 1+\epsilon$,
and $\sum_{i=1}^m w_i\leq (1+\epsilon)n$.
Therefore for all $i\in [m]$,
\[
\sigma_i(v)=a_i^{\top}(A^{\top}\diag(v)A)^{-1}a_i=a_i^{\top}\left(\frac n {\sum_{i=1}^m w_i}\cdot A^{\top}\diag(w)A\right)^{-1}a_i=\frac  {\sum_{i=1}^m w_i}n\sigma_i(w)\leq (1+\epsilon)^2.
\]
\end{proof}

From now on we focus on proving Theorem \ref{thm:good-approx}. 
Recall that $\phi_i(w)=\log \sigma_i(w)$ for $i\in [m]$.
Similar to Lemma \ref{lem:jenson},
we can prove the following lemma with the convexity of $\phi_i$.
\begin{lemma}[Telescoping]\label{lem:jl-jenson}
Fix $T$ as the number of main loops executed in Algorithm \ref{alg:faster}.
Let $w$ be the output at line \eqref{alg:output-jl} of Algorithm \ref{alg:faster}.
Then for $i\in [m]$,
\[
\phi_i(w)\leq \frac 1 T\log \frac m n+\frac 1 T\sum_{k=1}^{T}\log \frac{\hat w^{(k)}_i}{w^{(k)}_i}.
\]
\end{lemma}
\begin{proof}
Recall that $w=\frac 1 T\sum_{k=1}^T w^{(k)}$.
By Lemma \ref{lem:convexity},
$\phi_i$ is convex,
so we can apply Jensen's inequality to obtain
\begin{align*}
\phi_i(w)
&\leq \frac 1 T\sum_{k=1}^T\phi_i(w^{(k)}))\quad\quad\quad &\text{Jensen's inequality}\\
&=\frac 1 T\sum_{k=1}^T\log \sigma_i(w^{(k)})\quad\quad\quad &\text{by definition of $\phi_i$ function } \\
&= \frac 1 T\sum_{k=1}^{T}\log \frac{\hat w^{(k+1)}_i}{w^{(k)}_i}  &\hat w^{(k+1)}_i=w^{(k)}_i\sigma_i(w^{(k)})\\
&= \frac 1 T\sum_{k=1}^{T}\log \frac{\hat w^{(k+1)}_i}{\hat w^{(k)}_i}\frac{\hat w^{(k)}_i}{w^{(k)}_i}\\
&=\frac 1 T(\sum_{k=1}^{T}\log \frac{\hat w^{(k+1)}_i}{\hat w^{(k)}_i}+\sum_{k=1}^{T}\log \frac{\hat w^{(k)}_i}{w^{(k)}_i})\\
&=\frac 1 T\log \frac{\hat w^{(T+1)}_i}{\hat w^{(1)}_i}+\frac 1 T\sum_{k=1}^{T}\log \frac{\hat w^{(k)}_i}{w^{(k)}_i}\\
&\leq \frac 1 T\log \frac m n+\frac 1 T\sum_{k=1}^{T}\log \frac{\hat w^{(k)}_i}{w^{(k)}_i}.\quad\quad\quad &\text{by Proposition \ref{prop:hatw_property} and the initialization of $w^{(1)}$}
\end{align*}
\end{proof}

From Lemma \ref{lem:jl-jenson},
we can bound the expectation of $\phi_i$ directly.
\begin{lemma}[Expectation of $\log \sigma_i$]\label{lem:expectation-phi}
If $s$ is even, then
\[
\E[\phi_i(w)]=\E\left[\log a_i^{\top}\left(\sum_j w_ja_ja_j^{\top}\right)^{-1}a_i\right]\leq \frac 1 T\log \frac m n+\frac 2 {s}.
\]
where the randomness is taken over the sketching matrices $\{S^{(k)}\}_{k=1}^{T-1}$.
\end{lemma}
\begin{proof}
Recall the update rule
\[
 w_{i}^{(k+1)}=\frac 1 s\cdot \|S^{(k)}B^{(k)}((B^{(k)})^{\top}B^{(k)})^{-1}(\sqrt{w_i^{(k)}}a_i)\|_2^2.
\]
Let $y^{(k)}_i=B^{(k-1)}((B^{(k-1)})^{\top}B^{(k-1)})^{-1}(\sqrt{w_i^{(k-1)}}a_i)$ be a vector of size $m$.
Then $\hat w^{(k)}_i=\|y^{(k)}_i\|_2^2$,
and $w^{(k)}_i=\frac 1 s \|S^{(k)}y^{(k)}_i\|_2^2$,
 where each entry of $S^{(k)}$ is chosen i.i.d from $N(0,1)$.

Fix $y^{(k)}_i$.
Let us consider the distribution of $w^{(k)}_i$.
We first consider 1 coordinate of $S^{(k)}y^{(k)}_i$,
which is  
$
(S^{(k)}y^{(k)}_i)_j=\sum_{t=1}^m S^{(k)}_{jt}(y^{(k)}_i)_t
$.
Since each $S^{(k)}_{jt}$ is chosen from $N(0,1)$,
$S^{(k)}_{jt}(y^{(k)}_i)_t$ follows the distribution $N(0,(y^{(k)}_i)_t^2)$,
and $(S^{(k)}y^{(k)}_i)_j$ follows the distribution $N(0,\sum_{t=1}^m(y^{(k)}_i)_t^2)=N(0,\|y^{(k)}_i\|_2^2)=\|y^{(k)}_i\|_2\cdot N(0,1)$.
Hence $w^{(k)}_i$ follows the distribution of $\frac 1 s\|y^{(k)}_i\|_2^2\cdot\chi^2(s)$
where $\chi^2(s)$ is $\chi^2$-distribution with $s$ degree of freedom.

Hence if we only consider the randomness of matrix $S^{(k)}$,
then we have
\begin{equation*}
\E_S\left[\log\frac{\hat w^{(k)}_i}{w^{(k)}_i}\right]=\E_{z\sim \chi^2(s)}\left[\log \frac{\|y^{(k)}_i\|_2^2}{\frac 1 s\|y^{(k)}_i\|_2^2\cdot z}\right]=\E_{z\sim \chi^2(s)}[\log \frac{s}{z}].
\end{equation*}
Hence by the pdf of the $\chi^2$-distribution, 
assuming $s$ is even,
we have that
\[
\E_{z\sim \chi^2(s)}[\log {z}]=\int_{0}^{\infty} \frac 1{2^{s/2}\Gamma(s/2)}x^{s/2-1}e^{-x/2}\log x\de x=\sum_{i=1}^{s/2-1}\frac 1 i-\gamma+\log 2.
\]
The last equation use 4.352-2 of the 5th edition of \cite{gradshteyn2014table}, and $\gamma$ is the Euler constant.
Hence
\[
\E_{z\sim \chi^2(s)}[\log \frac{s}{z}]=\log s-\left(\sum_{i=1}^{s/2-1}\frac 1 i-\gamma+\log 2\right)\leq \log s-(\log \frac s 2+\gamma-\gamma+\log 2-\frac 2 s)=\frac 2 s.
\]
by the fact that $\sum_{i=1}^k \frac 1 i\geq \log k+\gamma$.

Therefore we have that
\[
\E[\phi_i(w)]\leq \frac 1 T\log \frac m n+\frac 2 {s}.
\]
\end{proof}

Since $\phi_i(w)=\log \sigma_i(w)$,
Lemma \ref{lem:expectation-phi} already provides some concentration results on $\sigma_i(w)$.
We can also prove stronger type of concentration by bounding the moments of $\sigma_i(w)$ directly.
\begin{lemma}[Moments of $\sigma_i$]\label{lem:expectation-sigma-moment}
For $\alpha>0$,
if $\frac s 2>\frac {\alpha} { T}$,
then
\[
\E[\sigma_i(w)^{\alpha}]=\E\left[(a_i^{\top}(\sum_j w_ja_ja_j^{\top})^{-1}a_i)^{\alpha}\right]\leq (\frac m n)^{\frac {\alpha} T}\cdot (1+\frac {2\alpha} {sT-2{\alpha}})^T.
\]
\end{lemma}

\begin{proof}
By Lemma \ref{lem:jl-jenson} we have
\[
\sigma_i(w)^{\alpha}\leq (\frac m n)^{\frac {\alpha} T}\cdot \prod_{k=1}^{T}(\frac{\hat w^{(k)}}{w^{(k)}})^{\frac {\alpha} T}.
\]
Fix $k$ and $\hat w^{(k)}$.
Similarly as proof of Lemma \ref{lem:expectation-phi},
for the moment let us only consider the randomness of $S^{(k)}$,
then we have
\[
\begin{split}
\E_{S^{(k)}}\left[(\frac{\hat w^{(k)}}{w^{(k)}})^{\frac{\alpha} T}\right]=\E_{z\sim \chi^2(s)}\left[(\frac{s}{z})^{\frac {\alpha} T}\right].
\end{split}
\]
Hence we have
\begin{align*}
\E_{S^{(k)}}\left[(\frac{\hat w^{(k)}}{w^{(k)}})^{\frac {\alpha} T}\right]&=\int_{0}^{\infty} (\frac s x)^{\frac {\alpha} T}\cdot \frac 1{2^{s/2}\Gamma(s/2)}x^{s/2-1}e^{-x/2} \de x\\
&=\frac{s^{\frac {\alpha} T}}{2^{s/2}\Gamma(s/2)}\int_{0}^{\infty}x^{\frac s 2-\frac {\alpha} T-1}e^{-x/2} \de x\\
&=\frac{s^{\frac {\alpha} T}}{2^{s/2}\Gamma(s/2)}\cdot \frac{\Gamma(s/2-\alpha/T)}{(1/ 2 )^{s/2-\alpha/T}}\\
&=(s/2)^{\frac {\alpha} T}\cdot \frac{\Gamma(s/2-\alpha/T)}{s/2}.
\end{align*}
where the third line uses 3.381-4 in \cite{gradshteyn2014table} and the condition that $\frac s 2>\frac {\alpha} T$.

By Lemma \ref{lem:gamma-function},
\[
\frac{\Gamma(\frac s 2)}{\Gamma(\frac s 2-\frac {\alpha} T)}\geq \frac{(\frac s 2-\frac {\alpha} T)}{(\frac s 2)^{1-\frac {\alpha} T}}~,
\]
which gives us
\[
\E_{S^{(k)}}\left[(\frac{\hat w^{(k)}}{w^{(k)}})^{\frac {\alpha} T}\right]\leq (\frac s 2)^{\frac {\alpha} T}\cdot \frac{(\frac s 2)^{1-\frac {\alpha} T}}{(\frac s 2-\frac {\alpha} T)}=\frac {\frac s 2}{\frac s 2-\frac {\alpha} T}.
\]
Because each $S^{(k)}$ matrix is independent to each other,
we have
\[
\E[\sigma_i(w)^{\alpha}]\leq (\frac m n)^{\frac {\alpha} T}\cdot (\frac {\frac s 2}{\frac s 2-\frac {\alpha} T})^T=(\frac m n)^{\frac {\alpha} T}\cdot (1+\frac {2{\alpha}} {sT-2{\alpha}})^T.
\]
\end{proof}

Now we are ready to prove Theorem \ref{thm:good-approx}.
\begin{proof}[Proof of Theorem \ref{thm:good-approx}]
Set $\alpha=\frac{3}{\epsilon}\log \frac {m}{\delta}$.
We can verify that
\begin{equation}\label{eq:alpha-size}
\alpha \geq \frac{\log ({m}/{\delta})}{\log \frac{1+\epsilon}{1+\epsilon/4}}.
\end{equation}
Notice that in this setting,
we have $sT\geq 4\alpha$.
Therefore,
for $i\in [m]$,
by Markov's inequality,
we have
\begin{align*}
\Pr[\sigma_i(w)\geq 1+\epsilon]
&=\Pr[\sigma_i(w)^{\alpha}\geq (1+\epsilon)^{\alpha}]\\
&\leq \frac{\E[\sigma_i(w)^{\alpha}]}{(1+\epsilon)^{\alpha}}\\
&\leq \frac{(\frac m n)^{\frac {\alpha} T}\cdot (1+\frac {2{\alpha}} {sT-2{\alpha}})^T}{(1+\epsilon)^{\alpha}}\quad\quad\quad &\text{by Lemma \ref{lem:expectation-sigma-moment}}\\
&\leq \frac{(\frac m n)^{\frac {\alpha} T}\cdot (1+\frac {2{\alpha}} {sT/2})^T}{(1+\epsilon)^{\alpha}}\quad\quad\quad &\text{because $sT\geq 4\alpha$}\\
&\leq \frac{(\frac m n)^{\frac {\alpha} T} e^{\frac {4\alpha} s}}{(1+\epsilon)^{\alpha}} \quad\quad\quad &\text{Here we use $1+x\leq e^x$}\\
\end{align*}
With our choice of $s,T$,
we can check that for sufficiently large $m,n$,
\[
(\frac m n)^{\frac {1} T}=(\frac m n)^{\frac{{\epsilon}/{10}}{\log  (m/\delta)}}\leq 1+{\epsilon}/{10},
\]
and
\[
e^{\frac {4} s}=e^{\frac {\epsilon}{20}}\leq 1+{\epsilon}/{10}.
\]
Hence
\[
\Pr[\sigma_i(w)\geq 1+\epsilon]\leq \left(\frac{(1+{\epsilon}/{10})^2}{1+\epsilon}\right)^{\alpha}\leq \left(\frac{1+{\epsilon} /4}{1+\epsilon}\right)^{\alpha}.
\]
Then by \eqref{eq:alpha-size},
we have
\[
\Pr[\sigma_i(w)\geq 1+\epsilon]\leq\frac \delta m.
\]
By union bound,
we have that 
\[
\Pr[\exists i\in[m], \sigma_i(w)\geq 1+\epsilon]\leq \delta.
\]

Then let us prove the second part in Theorem \ref{thm:good-approx}.
Fix $k$.
Recall that $B^{(k)}=\sqrt{W^{(k)}}A$.
Let $D^{(k)}$ be defined as \begin{align*}D^{(k)}:=B^{(k)}((B^{(k)})^{\top}B^{(k)})^{-1}(B^{(k)})^{\top},\end{align*}
then we can check that $D^{(k)}$ is an orthogonal projection matrix,
because
\[
(D^{(k)})^2=\left(B^{(k)}((B^{(k)})^{\top}B^{(k)})^{-1}(B^{(k)})^{\top}\right)\cdot\left(B^{(k)}((B^{(k)})^{\top}B^{(k)})^{-1}(B^{(k)})^{\top}\right)=D^{(k)}.
\]
Since $\mathsf{rank}(D^{(k)})=n$,
we can diagonalize $D^{(k)}$ as $D^{(k)}=\Lambda^{-1} E_n \Lambda$ where $\Lambda$ is an $m\times m$ orthogonal matrix,
and $E_n\in \R^{m\times m}$ is a diagonal matrix where the first $n$ diagonal entries are 1 and all the other entries are 0.
So we can rewrite the update rule as
\[
\begin{split}
w_{i}^{(k+1)}&=\frac 1 s\cdot \|S^{(k)}B^{(k)}((B^{(k)})^{\top}B^{(k)})^{-1}(\sqrt{w_i^{(k)}}a_i)\|_2^2\\
&=\frac 1 s\cdot\left((S^{(k)}D^{(k)})^{\top}(S^{(k)}D^{(k)})\right)_{ii}.
\end{split}
\]
Therefore
\begin{align*}
\sum_{i=1}^m w_{i}^{(k+1)}
=&\frac 1 s\cdot\sum_{i=1}^m\left((S^{(k)}D^{(k)})^{\top}(S^{(k)}D^{(k)})\right)_{ii}\\
=&\frac 1 s\cdot \Tr\left[((S^{(k)}D^{(k)})^{\top}(S^{(k)}D^{(k)})\right]\\
=&\frac 1 s\cdot \Tr\left[(D^{(k)})^{\top}(S^{(k)})^{\top}S^{(k)}D^{(k)}\right]\\
=&\frac 1 s\cdot \Tr\left[(S^{(k)})^{\top}S^{(k)}(D^{(k)})^2\right]\quad\quad\quad &\text{$D^{(k)}$ is symmetric}\\
=&\frac 1 s\cdot \Tr\left[(S^{(k)})^{\top}S^{(k)}D^{(k)}\right]\quad\quad\quad &\text{$(D^{(k)})^2=D^{(k)}$ }\\
=&\frac 1 s\cdot \Tr\left[S^{(k)}\Lambda^{-1} E_n \Lambda(S^{(k)})^{\top}\right].\quad\quad\quad &\text{diagonalization of $D^{(k)}$ }\\
\end{align*}

Let $\hat S^{(k)}=S^{(k)}\Lambda^{-1}$.
Here $\Lambda$ depends on previous randomness.
Notice that Gaussian distribution is invariant under orthogonal transform,
and $S^{(k)}$ is independent to previous $S$ matrices.
We conclude that $\hat S^{(k)}$ also has i.i.d entries that follows the distribution $N(0,1)$,
and $\hat S^{(k)}$ is independent to previous randomness.

So we have
\begin{align*}
\sum_{i=1}^m w_{i}^{(k+1)}=&\frac 1 s\cdot \Tr[\hat S^{(k)} E_n (\hat S^{(k)})^{\top}]
=\frac 1 s \sum_{p=1}^s\sum_{q=1}^n  (\hat S^{(k)})_{pq}^2.
\end{align*}
Namely, the distribution of $\sum_{i=1}^m w_{i}^{(k+1)}$ is $\frac 1 s \chi^2(ns)$.
Because for different $k$,
the randomness are independent,
we have $\sum_{i=1}^m w_i=\frac 1 T (\sum_{k=1}^T\sum_{i=1}^m w_i^{(k)})$ follows the distribution $\frac 1 {sT}\chi^2(nsT)$.
So we can set $t=\frac 1 {16}\epsilon^2\cdot nsT=\Theta(n\log \frac m \delta)$ in Lemma \ref{lem:chi-square-tail} to see that for sufficiently large $m,n$,

\begin{align*}
\Pr\left[\sum_{i=1}^m w_{i}^{(k+1)}\geq (1+\epsilon)n\right]=&\Pr_{z\sim \frac 1 {sT}\chi^2(nsT)}\left[z\geq (1+\epsilon)n\right]&\text{set $z=\sum_{i=1}^m w_{i}^{(k+1)}$}\\
=&\Pr_{z\sim \chi^2(nsT)}\left[z\geq (1+\epsilon)nsT\right]&\text{rescale $z$}\\
=&\Pr_{z\sim \chi^2(nsT)}\left[z-nsT\geq \epsilon\cdot nsT\right]\\
\leq &\Pr_{z\sim \chi^2(nsT)}\left[z-nsT\geq 2t+2\sqrt{nsT\cdot t}\right]\\
\leq & e^{-t} \leq \delta. &\text{by Lemma \ref{lem:chi-square-tail}}
\end{align*}
\end{proof}

\subsection{Runtime analysis}
We now analyze the running time of Algorithm \ref{alg:faster}.
The main loop is executed $T$ times,
and inside each loop,
we can first use $O(s\cdot \nnz(B^{(k)}))$ time to compute $S^{(k)}B^{(k)}$,
then solve $s$ linear systems to compute $F^{(k)}:=S^{(k)}B^{(k)}((B^{(k)})^{\top}B^{(k)})^{-1}$. 
Finally we can compute all of $\|F^{(k)}\cdot (w_i^{k})^{1/2}\cdot a_i\|_2^2$ in $O(s\cdot \nnz(A))$ time by first computing matrix $S^{(k)}\cdot \sqrt{W^{(k)}}\cdot A^{\top}$.
Recall that $B^{(k)}=\sqrt{W^{(k)}}A$,
so $\nnz(B^{(k)})\leq \nnz(A)$.
Notice that it takes at least $\nnz(A)$ time to solve the linear system $A^\top W^{(x)}Ax=b$.
Together with Theorem \ref{thm:algorithm-performance},
this gives us 
\begin{theorem}[Restatement of Theorem \ref{thm:short-performance-jl}]
For all $\epsilon,\delta\in (0,1)$,
we can find a $(1+\epsilon)$-approximation of John Ellipsoid inside a symmetric  polytope within $O\Big(\frac {1} {\epsilon}\log \frac m {\delta}\Big)$ iterations with probability at least $1-\delta$.
Moreover,
each iteration involves solving $O(\frac 1 \epsilon)$ linear systems of the form $A^{\top}WAx=b$ for some diagonal matrix $W$.
\end{theorem}


\newpage
\section{Matlab Code}\label{sec:code}
\lstset{language=Matlab,%
	style = Matlab-editor,
	basicstyle=\ttfamily\footnotesize,
    breaklines=false,%
    morekeywords={matlab2tikz},
    keywordstyle=\color{blue},%
    morekeywords=[2]{1}, keywordstyle=[2]{\color{black}},
    identifierstyle=\color{black},%
    stringstyle=\color{mylilas},
    commentstyle=\color{mygreen},%
    showstringspaces=false,
    numbers=none,%
    numberstyle={\tiny \color{black}},
    numbersep=9pt, 
}


\begin{lstlisting}[language=Matlab]
% Compute the John ellipsoid weight for a full rank matrix A
function [w_avg, iter] = FixedPoint(A,tol,w)
m = size(A,1); n = size(A,2);
exactTime = Inf; useJL = false;
if ~exist('tol', 'var'), tol = 0.01; end
if ~exist('w', 'var'), w = ones(m,1)*n/m; end
if issparse(A), A = A(:,colamd(A)); useJL = true; end

w_avg = w;
for iter = 1:ceil(10*log(m/n)/tol)
    tau = computeTau(w, 10+iter); %increase the JLdim 
    
    if max(tau./w) < 1 + tol
        w_avg = w; %use only the current w if we detect it has convgerged
        break;
    end
    w = tau; %update w
    
    if useJL
        w_avg = (1-1/iter)*w_avg + w/iter;
        if randi(iter) == 1 %with prob. 1/iter, test if we have converged
            tau_avg = computeTau(w_avg, 10+iter^2); %check conv of running avg
            if max(tau_avg./w_avg) < 1 + tol, break; end 
        end
    end
end

% compute tau = diag(B*inv(B'*B)*B') with B = sqrt(W) A
function tau = computeTau(w, JLdim)
    tstart = tic; %record times to decide when to use sketching
    B = spdiags(sqrt(w),0,m,m)*A;
    R = chol(B' * B); % In rare cases, chol fails due to numerical error

    if nnz(R) > n^2/20, A = full(A); useJL = false; end % use full if dense
    
	% remove if too slow
    if (exactTime == Inf || ~useJL) % use JL only if it is much faster
        S = R' \ B';
        tau = full(sum(S.^2,1)'); % compute tau exactly
        exactTime = toc(tstart);
    else
        S = B * (R\randn(n, JLdim));
        tau = full(sum(S.^2,2)/JLdim);  % compute tau using JL 
        JLTime = toc(tstart);
        if JLTime*(10+iter) > exactTime, useJL = false; end
    end
end
end
\end{lstlisting}

\end{document}